\documentclass[10pt, conference, compsocconf]{IEEEtran}
\IEEEoverridecommandlockouts

\usepackage{amsmath,amssymb,amsfonts}
\usepackage{algorithmic}
\usepackage{graphicx}
\usepackage{textcomp}
\usepackage{amsthm}
\usepackage{mathabx}
\usepackage{xspace}
\usepackage{xcolor}
\usepackage{mathrsfs}
\usepackage{cite}
\usepackage{url}
\usepackage{mdframed}
\usepackage{todonotes}
\usepackage{paralist}
\usepackage{authblk}
\usepackage{hyperref}

\newtheorem{theorem}{Theorem}
\newtheorem{definition}{Definition}
\newtheorem{assumption}{Assumption}
\newtheorem{construction}{Construction}

\newcommand{\noteJT}[1]{{$\langle${\textcolor{blue}{\textbf{#1}}}$\rangle$}}

\def\BibTeX{{\rm B\kern-.05em{\sc i\kern-.025em b}\kern-.08em
    T\kern-.1667em\lower.7ex\hbox{E}\kern-.125emX}}

\usepackage{pgfplotstable}
\usepgfplotslibrary{colorbrewer}
\usepgfplotslibrary{statistics}
\pgfplotstableread{%
Author min avg max median deviation
vgg-flowers 29.23823714 34.95312095 48.13869691 32.80741441 6.994066727 
dtd 52.85288 54.46467743 55.74429822 53.20353198 1.300885948
aircraft 93.47262621 102.4604299 111.609297 98.30078197 9.149221996 
ucf101 205.066431 243.9880307 401.2114701 214.0686505 65.61874725 
omniglot 485.7506762 531.987202 785.017925 501.9155896 89.75474286 
daimlerpedcls 655.336596 738.2630707 952.515991 702.2982916 106.4956777 
gtsrb 832.6738811 913.7884951 1068.475034 894.6429291 71.63945708  
cifar100 1075.773106 1194.053823 1408.078915 1169.566699 105.2741374 
svhn 1218.773486 1472.031525 2242.253265 1422.299515 314.6171453 
}\datatable

\pgfplotstableread{%
Author min avg max median deviation
vgg-flowers 4.615067959 5.976699209 11.05696321 5.374676943	1.855498516
dtd 26.26962805 29.90673053 32.73623991 30.30396056 2.201179286
aircraft 17.30487895 19.14774678 22.86621404 18.32386601 1.982724316
ufc101 65.3053 95.51870725 212.4429359 72.95532644 48.17843538
omniglot 468.6417241 618.9965071 989.7377579 566.0564494 176.9832026
daimlerpedcls 239.3759398 309.6753282 373.1835558 316.1552455 40.24977915
cifar100 4024.42163 5118.510823 6947.461023 4691.007913 1166.421668
svhn 1168.194177 2564.94285 5912.951988 1818.458548 1594.925383
gtsrb 1943.657319 2324.252596 3179.025302 2172.435562 442.5127481

}\indextable

\begin{document}
\pgfplotsset{compat=newest,compat/show suggested version=false}
\pagestyle{plain}

\title{Secure Single-Server Nearly-Identical Image Deduplication}

\author[1]{Jonathan Takeshita}
\author[1]{Ryan Karl}
\author[1]{Taeho Jung}
\affil[1]{Department of Computer Science and Engineering\\
University of Notre Dame\\
Notre Dame, IN 46556
}
\affil[ ]{\textit {\{\href{mailto:jtakeshi@nd.edu}{jtakeshi},\href{mailto:rkarl@nd.edu}{rkarl},\href{mailto:tjung@nd.edu}{tjung}\}@nd.edu}}

\flushbottom
\maketitle

\newif\ifframework

\newcommand{\calO}{\ensuremath{\mathcal{O}}}

\begin{abstract}
Cloud computing is often utilized for file storage. Clients of cloud storage services want to ensure the privacy of their data, and both clients and servers want to use as little storage as possible. Cross-user deduplication is one method to reduce the amount of storage a server uses. Deduplication and privacy are naturally conflicting goals, especially for nearly-identical (``fuzzy'') deduplication, as some information about the data must be used to perform deduplication. Prior solutions thus utilize multiple servers, or only function for exact deduplication. 
In this paper, we present a single-server protocol for cross-user nearly-identical deduplication based on secure locality-sensitive hashing (SLSH).
We formally define our ideal security, and rigorously prove our protocol secure against fully malicious, colluding adversaries with a proof by simulation. We show experimentally that the individual parts of the protocol are computationally feasible, and further discuss practical issues of security and efficiency.
\end{abstract}


\section{Introduction}

Cloud-based storage has become an increasingly popular solution for storing large amounts of data. Both users and providers of these systems have the common incentive to reduce the amount of storage and bandwidth these systems require. Users also have the incentive of privacy - they prefer for the provider and for other users to learn as little about their data as possible. The obvious solution to this problem is encryption - instead of uploading their files to a cloud server, users will instead upload an encryption of their file. 
Data encryption is neccesary to protect against data breaches, which may cost cloud storage providers millions of dollars in damages and lost business \cite{Target}.

As the amount of data stored by cloud storage providers increases, they will seek to mitigate their increasing costs from the extra storage. One technique to save storage and bandwidth is deduplication, where identical or similar pieces of data are detected, allowing servers to avoid storing redundant data. When identical or nearly-identical files are uploaded, the server will keep pointers to a single copy of data instead of storing redundant copies. There is a natural dissonance between deduplication and privacy. For accurate deduplication, some information about the file must be provided in order to test whether that file is similar to previously uploaded files. However, this provision of information defeats the purpose of encrypting files for data privacy, leading us to consider the question of to what extent a deduplication protocol can be both accurate and secure. 


\begin{figure}[t]
\centering
\includegraphics[width=\linewidth]{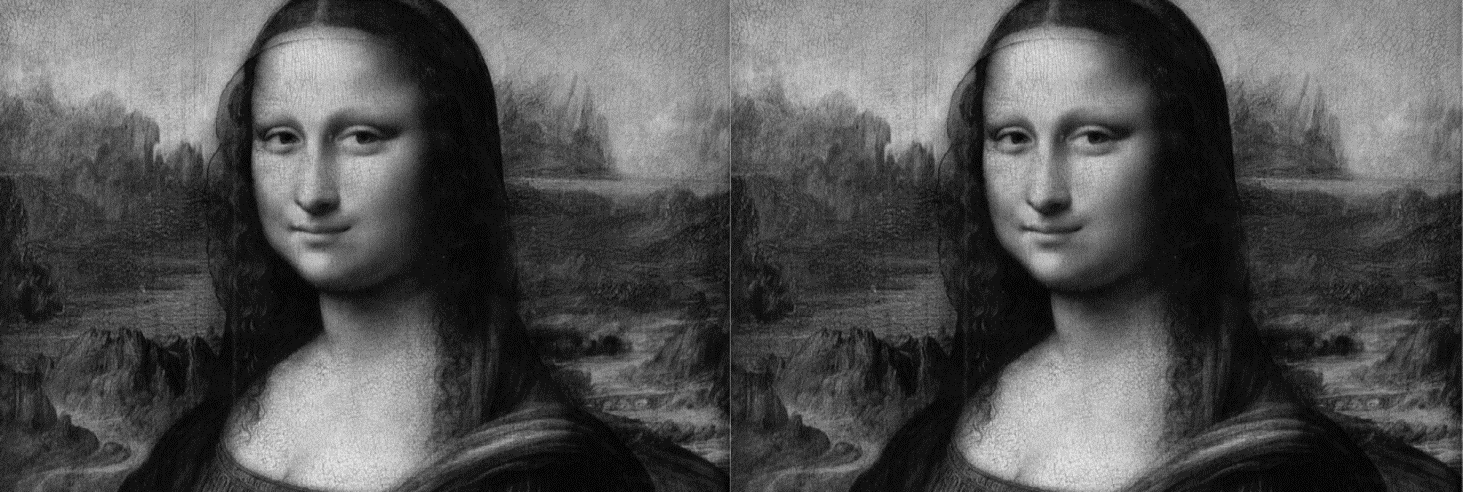}\vspace{-10pt}
\caption{Example of Nearly-Identical Images \cite{ML}}
\label{fig:mona_lisa}\vspace{-10pt}
\end{figure}

\noindent \textbf{Overview of Deduplication:}
Deduplication is the process of detecting identical or nearly-identical data for the purpose of conserving storage by storing unique data only.
Deduplication can take place on entire files, or on individual blocks of files, but it has been noted that the distinction is not important when considering deduplication schemes \cite{Liu2015}. 
Deduplication schemes can be classified as exact or nearly-identical. Exact deduplication works to determine if files are exact copies \cite{Yan2015,Lei2014,WEN2014,RASHID201654,Liu2015}. 
Nearly-identical deduplication works to detect highly similar files \cite{Chen2013AHD,Li2015}, in addition to exactly identical  files. However, this additional functionality requires more computation. We consider similarity as it relates to human perception of nearly-identical images as other similar works do \cite{Li2015} (e.g., Fig. \ref{fig:mona_lisa}).

In a cloud storage system utilizing deduplication for saving storage, the deduplication can be carried out by the clients or the server. It is often preferable in high-trust scenarios for the server to carry out deduplication, to reduce the computational load on the clients. However, in situations where privacy is a concern, clients may not wish to provide the server with the necessary data to perform deduplication.
Deduplication can be performed between data from multiple users or only across data from a single user. Only applying deduplication on a per-user basis is a simple answer to concerns of cross-user privacy, but cannot reduce storage in the event of multiple users storing the same file.

Client-based secure nearly-identical deduplication is most useful in a scenario where clients' computation is plentiful, but their storage is limited. For example, ordinary smartphones can perform the computation needed to carry out nearly-identical deduplication when plugged in at night, and this deduplication can reduce the use of smartphones' limited storage.  This is also applicable with use cases involving IoT devices. For instance, there has been recent interest in utilizing IoT devices to allow for the affordable deployment of biometric technology, but a major challenge in this scenario is building efficient systems in spite of the space constraints \cite{dhillon2017lightweight}.  Secure deduplication can be used to decrease the need to store a large amount of redundant data on a server, which would alleviate practical space constraints when leveraging such technology in the wild.

\noindent \textbf{Summary of contributions:} (1) A review of related work in the area of deduplication; (2) Design and implementation of a nearly-identical deduplication scheme for images, with \textit{security against fully malicious, colluding adversaries and only utilizing a single untrusted server}; (3) A proof of security of our scheme, with a discussion of practical issues; (4) Experiments with real-world datasets showing the feasibility of our protocol and implementations.

\section{Related Work}

\subsection{Exact Two-server Deduplication}

In some schemes, hashing is used to protect data privacy. The scheme proposed by Wen et al. uses two servers to construct a system for exact deduplication \cite{WEN2014}. A storage server will store both hashes and encryptions of users' images, while a verification server will store only the hashes. The multiple redundancy of both the storage and verification server storing image hashes protects the user in the case that one server behaves maliciously. Convergent encryption is used to ensure users can access deduplicated images.
A scheme proposed by Yan et al. uses proxy re-encryption to share data between users who have attempted to upload identical data \cite{Yan2015}. Similarly to the work of Wen et al., a verification server is used to store information needed for deduplication.

\subsection{Exact One-Server Deduplication}


The scheme of Rashid et al. similarly leverages hash values for image privacy, but with only one server \cite{RASHID201654}. 
Beyond the storage saved by using deduplication, this scheme achieves even better savings by compressing images. The compression takes a tree-like, hierarchical form, where the original image cannot be reliably reconstructed without the most significant information from the higher levels of the tree. Thus, by only encrypting the most significant information from the compression of an image, the amount of encrypted data sent and stored can be reduced, saving bandwidth and storage.

Liu et al. constructed a system that allows secure deduplication with only one central server \cite{Liu2015}. A key building block of this protocol is user-based key sharing, which takes place through a subprotocol known as Password-Authenticated Key Exchange (PAKE) \cite{abdalla2005simple}. 
In this protocol, upon a file upload the server will compare a short hash (e.g. 13 bits) of the file to short hashes of previously uploaded files, and use this to construct a shortlist of users that may have previously uploaded identical files. Data privacy is preserved because many collisions (of different files) are intentionally created in this list. Additional computation by both the clients and server allow the server to check whether a duplicate file exists. If it does, the server will return an encryption key of the file to the uploader. If not, then the server will accept the file as a unique one. 

The protocol is provably secure against malicious and colluding adversaries. This protocol also has the advantage of being generalized to any type of data, not just images or text. Many practical attacks are precluded by the use of server- or client-side rate limiting.
The protocol does have room for improvement. Its utility is strictly limited to the scenario of exact deduplication, because PAKE requires exact equality of the parties' inputs for identical key exchange. From an efficiency viewpoint, the protocol requires up to six communication rounds per upload.  

\subsection{Nearly-identical Two-Server Deduplication}

By using two servers, Li et al. are able to construct a system for secure nearly-identical image deduplication\cite{Li2015}. Their protocol uses one server for deduplication, which stores the perceptual hashes of users' images and performs the work of deduplication. The other server stores the users' encrypted images. A perceptual hashing method is used to perform image deduplication by mapping similar images to identical hashes.
In this protocol, the deduplication server is only able to see perceptual hashes of the users' images, and the storage server sees only encryptions of those images, making this system effective for protecting users' privacy against other parties or external adversaries. However, the system has users share group keys among themselves, which requires that users will know \textit{a priori} whom will be uploading similar images. Thus if two users in different groups upload identical or similar images, the storage server will store both. Later work extended this system with Proof of Ownership and Proof of Retrieval \cite{bini2018proof}.

\subsection{Proof of Ownership/Retrieval}
Proof of Ownership/Retrieval (PoW and PoR) schemes aim to provably ensure a client's ownership of a file or their ability to recover a stored file from a server, respectively \cite{juels2007pors,bowers2009proofs}. Both of these concepts have been applied to deduplication, especially PoW \cite{gonzalez2015efficient,pow-boosting,bini2018proof,zheng2012secure,blasco2014tunable,jin2013anonymous,yu2015proof}. PoW and PoR have even been applied to secure nearly-identical deduplication, though the scenario is much less adversarial than ours \cite{bini2018proof}. PoR is perpendicular to our work: though it could be applied with our scheme, it is not the focus of this work. We use PoW in our work for both deduplication and access control.

\section{Security Definitions and Ideal Functionality}

\subsection{System Model}

We consider the scenario where the parties consist of an arbitrary number of users and a single cloud storage server. The users wish to use the server to securely store images, but without allowing the server to learn the content of their images, or the other users being able to determine the content of their images unless the users both have an identical or nearly-identical image. All parties have the shared goals of wishing to conserve storage while also keeping their own information secure. We thus only consider the case where the server stores encrypted images.

\subsection{Adversary Model and Goals}

We consider the (very challenging) case of \textit{fully malicious, colluding} adversaries. The server, any of the clients, or any collusion thereof may take any action. Their adversaries' goal in this scenario is to gain some semantically useful information about an innocent users' data that they do not already possess. They may also choose to take actions that may abort the correct execution of the system (e.g. refusing to reply, sending junk data). This type of behavior is a practical issue, and does not compromise the privacy of innocent users' data.

\subsection{Security Model}

We define the ideal functionality $\delta$ of secure nearly-identical  deduplication over encrypted data in Fig. \ref{fig:ideal-model}. This functionality is ideal in the sense that an ideal, fully trusted `system' takes the input and returns the output, without disclosing any information to any participant. This functionality characterizes the views of adversaries in an ideal world where the whole process is delegated to an ideal `system'. Our protocol (Section \ref{sec:protocol}) will be designed such that the adversaries' views during the execution of it in the real world are computationally indistinguishable from the adversaries' views in the ideal world. The three types of participants are the storage server $S$, the user $U_i$ attempting to upload an image, and preexisting users $U_j$ who have already uploaded a file. A protocol implementing $\delta$ is considered secure if it implements $\delta$ and leaks negligible information about $U_i$'s or $U_j$'s images and keys, and $S$ only knows whether a PAKE transaction has been initiated between two parties or not. 
An adversary $A$ may compromise any one of $S, U_i$, or $U_j$. We highlight that, in $\delta$, $U_j'$ learns nothing about $I_i'$ if $I_i'$ and $I_j'$ are not similar, and the server $S'$ learns only some data relevant to $I_i'$ that cannot be used to reconstruct $I_i'$ except with negligible probability (e.g., hashes of an image or of its feature vector).



We follow the approach of \cite{hazay2008efficient} to formalize this intuition in the following definition:
\begin{definition}
Let $\Gamma$ and $\delta$ be the real and ideal functionalities respectively. Protocol $\Gamma$ is said to securely compute in the presence of fully malicious adversaries with abort if for every non-uniform probabilistic polynomial time adversary $A$, for the real model there exists a non-uniform probabilistic polynomial-time adversary $S$ for the ideal model such that for every input $x, x' \in \{0, 1\}^*$ with $|x| = |x'|$, security parameter $\kappa$, every auxiliary parameter input $z \in \{0, 1\}^*$, and locality sensitive hashes $h_{p_{z}}(x), h_{p_{z}}(x')$, the views generated by  $\{IDEAL_{\delta,S(z)}(x', k, h_{p_{z}}(x')), \}$ and $\{REAL_{\Gamma,A(z)}(x, k, h_{p_{z}}(x))\}$ are computationally indistinguishable \textit{w.r.t.} $\kappa$. If users fail to respond during the protocol, the protocol aborts.
\end{definition}

Here, the inputs $x, x'$ are the images uploaded by users, the security parameter $\kappa$ is the number of bits of security, the auxiliary parameter data $z$ contains details about the implementation of the protocol (e.g. the cyclic group used in PAKE, which hash functions to use), and the locality-sensitive hashes $h_{p_{z}}(x), h_{p_{z}}(x')$ are the hashes of the two images $x, x'$ which indicate the similarity scores of the images. Note that some practical attacks are not prevented by this ideal functionality, most notably that adversaries are able to learn in a quantifiable way how similar their uploaded images are to those of another user. This is a common problem in secure deduplication schemes, and to the best of our knowledge there is no consensus in the community of how to address this concern \cite{Chen2013AHD, gang2015secure}. We address some of these attacks with practical safeguards discussed in Section \ref{sec:practical-problems}.

\begin{figure}[t] 
    \centering
    \begin{mdframed}
\textbf{System Inputs:}
\begin{itemize}
    \item For $0 \leq i \leq N$, uploader $U_i'$ inputs an image $I_i'$.
    \item For $0 \leq j \leq N$ and $j\neq i$, previous uploaders $U_j'$ input images $I_j'$. 
    \item The server $S'$ has encryptions of images $I_j'$ under symmetric keys $k_j'$
\end{itemize}
\textbf{System Outputs:}
\begin{itemize}
    \item If $I_i'$ is nearly identical to some uploaded $I_j'$, $U_i'$ gets $k_j'$ as well as the encrypted $I_j'$, and $S$, $U_i'$, and $U_j'$ may learn $i$ and $j$. Otherwise, $U_i'$ gets a new symmetric key $k_i'$, and $S'$ gets an encryption of $I_i'$ under $k_i'$.
\end{itemize}

\end{mdframed}\vspace{-10pt}
    \caption{Ideal Functionality $\delta$}
    \label{fig:ideal-model}\vspace{-10pt}
\end{figure}

\section{Design of a Concrete Protocol}\label{sec:protocol}


\subsection{Preliminaries}

\begin{definition}
A locality-sensitive hash scheme is a distribution on a family $F$ of hash functions operating on a collection of objects $K$, such that for two objects $x, y \in K$, $\textbf{Pr}_{h_p \in F} [h_p(x) = h_p(y)] = sim(x, y)$ for a hash parameter $p$
where $sim(x, y) \in [0, 1]$ is some similarity function defined
on $K$.
\end{definition}

Intuitively, a locality-sensitive hash scheme hashes similar objects to the same value. However, this definition makes no statements about the security of the function. In particular, the definition does not imply preimage resistance, meaning that an adversary may be able to reverse the locality-sensitive hash to find the original input. 

\begin{definition}
 A secure locality-sensitive hash (SLSH) is a locality-sensitive hash function $h$ that has the property of preimage resistance: for any input $x$ and a polynomially bounded number of parameters $p_1 \cdots p_t$, it is computationally intractable to find $x$ given only $h_{p_1}(x) \cdots h_{p_t}(x)$ and $p_1 \cdots p_t$.
\end{definition}

A SLSH can be constructed from the standard assumption of the existence of cryptographic hash functions \cite{KatzLindell}.

\begin{construction}
A SLSH can be constructed as the composition $H \circ LSH_p$ of a locality-sensitive hash function $LSH_p(x)$ and a cryptographic hash $H(x)$, i.e. $SLSH(x):=H(LSH_p(x))$.
\end{construction}

Cryptographic hash functions are one-way functions, with the property of preimage resistance. 
Using a cryptographic hash to construct a SLSH gives it the property of preimage resistance, which is desirable for our application.




\begin{definition}
 A password authenticated key exchange (PAKE \cite{abdalla2005simple}) is a functionality where two parties $P_1$ and $P_2$ each input a password $pw_1$ and $pw_2$ receive as respective output keys $k_1$ and $k_2$. If $pw_1 = pw_2$, then $k_1 = k_2$, and otherwise $P_1$ and $P_2$ cannot distinguish $k_1$ and $k_2$ respectively from a random string of the same length. 
\end{definition}

\subsection{Protocol Description}

Our protocol $\Gamma$ is shown in Fig.  \ref{fig:protocol}, where the parties consist of a single central server $S$ and $N$ users $U_{1}, \cdots U_{N}$. The server maintains $t$ hash tables $HT_{1} \cdots HT_{t}$ used in deduplication, and makes the parameters of each table public. When a user $U_{i}$ wishes to upload an image $I_{i}$ they will first calculate a feature vector of the image, $V_{I_i}$, and then find $t$ SLSHes $H_{1} \cdots H_{t}$ according to the server's hash parameters. 
After this client-side calculation, the uploading user $U_{i}$ will then send the SLSHes for its image $I_{i}$ to $S$. The server then constructs a shortlist of possibly similar images by checking the received hash values against the hashes in the tables $HT_{1} \cdots HT_{t}$ and noting any collisions. Images $I_i, I_j$ whose SLSHes collide will have similar feature vectors (i.e. $V_{I_{i}} \approx V_{I_{j}}$), and are similar. Thus the server can identify any images with at least $c$ (a scenario-dependent parameter) hash collisions as similar images.
If no other image is found to be similar to the new image $I_i$, then the server indexes each hash value $H_{x \in [1,t]}$ into table $HT_{x}$. It then allows the uploading user to upload an encryption of its image $ENC_k(I_{i})$ (with the encryption key $k$ being unique for $I_i$), which $S$ then stores.
If the image being uploaded $I_i$ is found similar to a stored image $I_j$ (so that $V_{I_{i}} \approx V_{I_{j}}$), then the server directs the original owner $U_{j}$ and new uploader $U_{i}$ to distribute the image's encryption key to $U_i$ through PAKE, and allows the new uploader to access the (encrypted) original image. (In case of multiple possible similar images, any of the nearly-identical images can be chosen as the similar one, though a salient choice would be to use the image with the most collisions.)



\begin{figure*}[htpb] 
    \begin{mdframed}
    
\begin{center}
   Our Protocol $\Gamma$\medskip
\end{center}\vspace{-5pt}

\noindent \textbf{Definitions:} Let $S$ be the central server, and $U_1 \cdots U_N$ be users of the server. The server has $t$ hash tables $HT_1 \cdots HT_t$ with $t$ sets of public parameters. The algorithms $GEN, ENC, DEC$ are the key-generation, encryption, and decryption algorithms of a symmetric-key encryption scheme. $\{H_p\}$ is a family of SLSH algorithms, where elements are parameterized for a set of parameters $p$.

\noindent \textbf{Client-Side Computation:}  

      \begin{enumerate}
          \item  User $U_{i}$ wishes to upload an image $I_{i}$ and calculates a feature vector of the image, $V_{I_i}$. The user then calculates $t$ SLSHes $H_{1}=H_{p_1}(V_{I_i}), \cdots, H_t=H_{p_t}(V_{I_i})$ of $V_{I_i}$ according to the server's hash parameters $p_1 \cdots p_t$.

        \item $U_{i}$ will then send the hashes $H_{x \in [1,t]}$ for its image $I_{i}$ to $S$. 
        
        \end{enumerate}
        
\noindent \textbf{Server-Side Deduplication:}

        \begin{enumerate}
            \item $S$ will then compare $H_{x}$ with values already in its hash tables. For $x \in [1,t]$, the server will check $H_x$ against the values (filenames of previously indexed images) stored in $HT_x$ at $H_x$, and add any values found to a shortlist, counting how many times that value has been found in the tables. Once this has been completed, the server can conclude that the image $I_i$ is similar to another image $I_{j \neq i}$ owned by a user $U_j$ if its hashes have at least $c$ collisions with the hashes from $I_j$.
            
            \item  If no other image is found to be similar to the new image, then the server indexes the filename of $I_i$ in its tables $HT_1 \cdots HT_t$ at the locations $H_1 \cdots H_t$, and allows the uploading user to upload an encryption of its image $ENC_k(I_{i})$, which $S$ then stores.
            
            \item If the image being uploaded is found similar to another stored image (i.e. $V_{I_i} \approx V_{I_{j}}$), then the server directs the original owner $U_{j}$ and new uploader $U_{i}$ to share the encryption key of $I_j$ through PAKE, and allows the new uploader to access the encryption of the original image.
            
        \end{enumerate}

\noindent \textbf{Client-Based Access Control}
        \begin{enumerate}
           
           \item After being so directed by $S$, $U_i$ and $U_j$ choose and share fresh sets of SLSH parameters $p_i, p_j$ respectively. They then calculate the hashes $H^i_i = H_{p_i}(I_i)$ and $H^i_j = H_{p_i}(I_j)$ of their feature vectors $V_{I_i}$ and $V_{I_j}$ according to $p_i$, and similarly calculate the hashes $H^j_i$ and $H^j_j$ according to $p_j$.
           
           \item $U_i$ and $U_j$ perform the PAKE protocol twice, using $H^i_i$ and $H^i_j$ respectively as input to the first session and $H^j_i$ and $H^j_j$ respectively as input to the second session. They receive back keys $k^i_i$ and $k^i_j$ respectively from the first session, and keys $k^j_i$ and $k^j_j$ respectively from the second session. They then concatenate their keys to form $k_i = k^i_i \| k^j_i$ and $k_j = k^i_j \| k^j_j$.
            
            \item If the users' images $I_i$ and $I_j$ are similar, then their feature vectors will be similar, and with high probability will be hashed to the same value under a SLSH. Then $k_i = k_j$, and decryption succeeds, allowing $U_i$ to recover $k$.
            
            \item If the images $I_i$ and $I_j$ are not similar, then $U_i$ cannot recover $k$ and will not be able to decrypt $I_j$.
        \end{enumerate}

\end{mdframed}\vspace{-10pt}
    \caption{\normalsize Our deduplication protocol $\Gamma$ with security against fully malicious, colluding adversaries.}
    \label{fig:protocol}\vspace{-10pt}
\end{figure*}

\noindent \textbf{Feature Extraction:}
For feature extraction we use the ResNet neural network architecture. Compared to other similar architectures for image feature extraction (e.g. the VGG and AlexNet architectures used by the system of Pintrest \cite{Pintrest}), ResNet can achieve higher accuracy with less computation, making it an attractive choice for accuracy and efficiency \cite{ResNet}. 

\noindent \textbf{Dimensionality Reduction:}
We use a well-known method of locality-sensitive hashing based on random planes for dimensionality reduction \cite{Charikar,grauman2007pyramid}.
To construct a SLSH from the LSH, we compose the locality-sensitive hash with a cryptographic hash function (SHA256 in our implementation).

\noindent \textbf{Nearest-Neighbor Search:}
The nearest-neighbor search is made easy by the SLSHing carried out previously. We use multiple hash tables to be robust against the small possiblity that similar items might differ in parts of their locality-sensitive hash (leading to a potentially wildly different SLSH value).
Items hashed to the same hash buckets will be similar, thus we can simply choose the item with the most hash collisions (above a minimal threshold) as a similar image.

\noindent \textbf{Access Control:}
For post-deduplication image sharing, we use the PAKE method \cite{abdalla2005simple}. After two users are notified to share keys by the server, they first mutually agree upon a new set of SLSH parameters. They then calculate SLSHes of the feature vectors of their images, and perform PAKE-based key sharing with those hashes as input. When the users' images are similar, the SLSHes used as input will be equal with high probability, and the users will receive identical keys. 
The key received by the holder of the original image is used to symmetrically encrypt the original image's encryption key. That encryption is then sent to the uploading user. If the keys received from PAKE are identical, then the uploading user will be able to later decrypt the encryption of the original image that the server stores. If the users' images are not similar, then the SLSHes of their feature vectors will be different (with high probability), and decryption of the encrypted encryption key will fail, because the PAKE protocol will return different keys to participants with differing inputs.

\subsection{Advantages of our Protocol}

Our system uses a single untrusted server, with the user performing feature extraction and dimensionality reduction before sending hash values to the server. The users also do the work in rehashing and PAKE required for access control. This arrangement allows for a very high degree of security and utility in a highly adversarial setting. While more computation must be done on the user's side, this is not prohibitively expensive. 

\section{Proof of Security}

\begin{theorem}
$\Gamma$ securely computes the ideal functionality $\delta$ in the presence of fully malicious, colluding adversaries with abort if PAKE is secure against fully malicious adversaries, the encryption scheme used is also secure, and cryptographic hash functions exist. If users fail to respond during the PAKE protocol, the protocol aborts.

\end{theorem}

\begin{proof}
We will show that the execution of the
protocol $\Gamma$ in the real world is computationally indistinguishable from the execution of the ideal functionality $\delta$. This proof is inspired by that of \cite{Liu2015}. The simulator $SIM$ can both access $\delta$ in the ideal model and obtain messages that the corrupt parties would send in the real model. $SIM$ generates a message transcript of the ideal model execution $\delta$ that is computationally indistinguishable from that of the real model execution $\Gamma$. To simplify the proof we assume that the PAKE protocol is implemented as an oracle to which the parties send inputs.

Our proof assumes that parties will send dishonestly constructed messages, and does not consider a party choosing to not send a message. 
Note that if any party refuses to respond or sends junk data, the honest parties can abort the protocol at that point, allowing us to achieve security with abort.

\noindent\textbf{A corrupt uploader $CU$:}
    We first assume that $S$ and $U_{j}$ are honest and construct a simulator for $CU$.  The simulator records $CU$'s SLSHes of the form $H_{p}(V_{CU})$.  After receiving a message $MSG_{CU,U_j}$ from $S$ indicating that $CU$ and a user $U_j$ have similar images, it records the calls that $CU$ makes to the PAKE protocol with $U_j$. Conversely, if no existing image stored on $S$ is similar to $I_{CU}$ for all other users $U_j$, this implies there will be no further communication between $CU$ and any other user.  If $CU$ uses a value $H(V_{I_{CU}})$ in that call that appears in a hash table $HT_x$, the simulator invokes $\delta$ with the image $I_j$ that corresponds to the hash $H(V_{I_{CU}})$. In this case, $CU$ will receive a key $k_{I_{CU}}$. 
    
    If an image $I_j$ similar to $I_{CU}$ has been uploaded by any $U_j$, $k_{I_j} = k_{I_{CU}}$ is the key corresponding to that image.  We now show that $\Gamma$ and $\delta$ are identically distributed. If $I_{CU}$ already exists in the server's storage and $CU$ behaves honestly, then $V_{I_{CU}} \approx V_{I_{U_j}}$ and thus $k_{CU} = k_{U_{j}}$.  If $I_{CU}$ does not already exist in the server's storage, then something encrypted by $k_{CU}$ will be indistinguishable from random by the security of the symmetric encryption scheme. Thus, $ENC_{k_{CU}}({I_{CU}})$ will be indistinguishable by $S$ from a random value. Now if $CU$ deviates from the protocol then the only action it can take, except for changing its input hash, is to replace its encryption of the image corresponding to $V_{I_{CU}}$ with an encryption of a different image or random data, that it then sends to $S$. 
    
    The result of both types of malicious behavior is that $CU$ sends $S$ hashes $H_1 \cdots H_t$ that are not correct SLSHes corresponding to the data $ENC_{k_{CU}}(I_{CU})$ uploaded. In this case, there are two possibilities: either the server will incorrectly not identify $I_{CU}$ as being similar to any stored image when it should, or the server will incorrectly identify $I_{CU}$ as being similar to some other image.
    
    In the first case, upon initial upload, as no similar images to $I_{CU}$ are identified, $CU$ does not exchange keys with any other user prior to upload, and learns nothing about another user's image. However, another user $U_j$ later uploading $ENC_{k_{I_j}}(I_j)$ may then have their image identified by $S$ as being similar to $I_{CU}$. In this case, the users will then make calls to the PAKE protocol. $CU$ cannot learn anything more than what is described in the security definition about $I_j$ from either $ENC_{k_{I_j}}(I_j)$ or from SLSHes of $V_{I_j}$. For $CU$ to learn anything about $I_j$, they need to recover $k_{I_j}$. However, without having $I_j$ or a highly similar image \textit{a priori}, $CU$ cannot correctly calculate new SLSHes, and thus cannot receive $k_{I_j}$ through PAKE. Thus in the first case, $CU$ cannot learn anything more than what is described in the security definition about $I_j$.
    
    In the second case, because an image in the server's storage is similar to the new image, $CU$ will begin the PAKE protocol with user $U_j$ who owns the similar image. $CU$ may or may not have honestly generated $H_1 \cdots H_t$ from an image $I_{CU}'$. If this was not the case, then as above $CU$ cannot recover $k_{I_j}$, and cannot learn anything more than what is described in the security definition about $I_j$.
    On the other hand, if $CU$ generated $H_1 \cdots H_t$ honestly from $I_{CU}'$, then $I_{CU}' \approx I_j$, and $CU$ is able to correctly generate new locality-sensitive hashes for its PAKE sessions with $U_j$. Then in this case, $CU$ can recover $k_{I_j}$, allowing it to download and decrypt $ENC_{k_{I_j}}(I_j)$, recovering $I_j$. However, because $I_{CU}' \approx I_j$, this does not violate ideal functionality or the security definition.
   
    We assume that $CU$ sends $q$ messages $m_1 \cdots m_q$ during its execution (hashes, etc.), and replaces $y$ of these messages.  In the real model $\Gamma$, the execution will change if there is an index $j$ such that the message $m_j$ in $\Gamma$ (which corresponds to the same $m_j'$ in $\delta$) is replaced by $CU$.
     As a result, $CU$ will change the execution even though it inputs a modified encrypted image or hash. The probability for this event is $y/q$, but it will be detected with high probability. However, in $\delta$, the same result will occur in the event that a replaced element is chosen by the simulator. The probability of this event occurring is also $y/q$ by the security of PAKE. Thus, we conclude that the views of $\Gamma$ and $\delta$ are identically distributed.

\noindent\textbf{A corrupt previous uploader $CP$: }
    Here, we say that $CP$ has previously been honest in uploading its hashes and encrypted image to the server. 
    $CP$ will learn from this execution if $H_{p'}(V_{I_i}) = H_{p'}(V_{I_{CP}})$, for $p' \in  \{p_i, p_j\}$. 
    The simulator $SIM$ will receive $CP$'s input $H_{p'}(V_{I_{CP}})$, but since $CP$ has previously uploaded $ENC_{k_{I_{CP}}}(I_{CP})$, it only needs to recover the key corresponding to $k_{CP}$.
    The simulator $SIM$ first checks whether the hashes $H_1 \cdots H_t$ of $V_{I_{i}}$ match the hashes of $I_{CP}$ in $S$'s hash tables. If not, $CP$ is not identified as having a similar image to $I_i$, and will take no action. Otherwise, $S$ observes $CP$'s inputs $H_{p'}(V_{I_{CP}})$ to the PAKE protocol, the new key $k_i$ that $U_i$ gains from PAKE, and the message $ENC_{k_{CP}}(k_{I_{CP}})$. 
    Then $CP$ and $U_i$ exchange the same information necessary to run the PAKE protocol as a black box. 
    The simulator checks if $H(V_{I_{CP}}) = H(V_{I_{i}})$. If so, it extracts and sends $k_{CP}$ to ${U_{i}}$. 
    
    To show that the simulation is accurate, note if $CP$ behaves honestly, then $\delta$ and $\Gamma$ are obviously indistinguishable. $CP$ can only deviate from the protocol in two ways.
    First, it can deviate from the PAKE protocol in a way that forces PAKE to abort, or by providing incorrect input so that the symmetric keys from PAKE do not match. In either case, though the adversary has managed to prevent the successful operation of the protocol, it has not learned any new information about other parties' images, due to the security of PAKE. Second, it can abide honestly by the PAKE protocol, but send an incorrect key that $U_i$ then cannot use to successfully recover $I_{CP}$. Again, $CP$ does not learn any new information about another party's image, and we can safely abort if necessary. We conclude that the views of $\delta$ and $\Gamma$ are identically distributed.

\noindent\textbf{A corrupt server $CS$:}
The simulator will first act as a user $U_i$ with image $I_i$, and send hash values $H_1 \cdots H_t$ to $CS$. The server $CS$ will query those values against its tables $HT_1 \cdots HT_t$, and either find that there is a user $U_j$ with a similar image, or that no similar image has been stored with the server. If the server behaves honestly in the second case or dishonestly in the first, then the server will accept the upload of $ENC_{k_{I_i}}(I_i)$. By the security of the symmetric encryption and the security of the SLSH, the server cannot learn any new information beyond what is described in the security definition about $I_i$ from $H_1 \cdots H_t$ and $ENC_{k_{I_i}}(I_i)$.

The server can also behave maliciously by telling $U_i$ that they have uploaded an image $I_i$ similar to an image $I_j$ previously uploaded by $U_j$, and directing them to perform PAKE to share keys. If this happens, then the users $U_i$ and $U_j$ will with overwhelming probability choose different passwords in their PAKE protocol, and will thus be unable to share encryption keys. Thus when $U_i$ and $U_j$ have different images, neither can learn anything more than what is described in the security definition about the other's image, even when $CS$ behaves dishonestly. Suppose the server has $m$ other images that they can choose to identify as similar with $I_i$. Deduplication fails if the owner $U_j$ and their image $I_j$ are not chosen correctly by $CS$, which happens in both the real and the ideal model with the same probability $r/m$ where $r$ is the number of dissimilar images. In both cases, $U_i$ and/or $U_j$ will be able to detect this behaviour with high probability. Thus $\delta$ and $\Gamma$ are identically distributed.



\noindent\textbf{Colluding corrupt server $CS$ and corrupt previous uploader $CP$:}

When the honest user $U_i$ uploads a new image, $CS$ can either behave honestly or maliciously. If $CS$ behaves honestly, then this reduces to the above case of a single corrupt previous uploader. If $CS$ does not, then it can take only one action not already enumerated in the above case of a single malicious server. The server can falsely claim that $I_i$ is similar to an image $I_{CP}$ owned by $CP$, and direct them to exchange keys. Then this reduces to the case of a single corrupt previous uploader.

\noindent\textbf{Colluding corrupt server $CS$ and corrupt uploader $CU$:}

Similarly, the only dishonest actions the collusion of $CS$ and $CU$ can take that differ from already-enumerated cases is for $CS$ to falsely tell an innocent previous uploader $U_j$ that $CS$ has attempted to upload an image similar to an image $I_j$ stored by $CS$. Then this also reduces to the case of a single corrupt uploader.

\end{proof}

\section{Practical Problems} \label{sec:practical-problems}


\subsection{Inference and Anonymity}

The protocol $\Gamma$ does not allow participants to learn anything more than what is described in the security definition about images unless they possess a similar image \textit{a priori}. However, an inference attack is trivial to mount: a user can easily learn if another user has uploaded an image by simply requesting to upload that image to the server. This attack can be prevented by making all connections anonymous, which can be accomplished through onion routing \cite{reed1998anonymous}. When the server notifies two users to share encryption keys, it then also gives them a one-time-use token pair that the users can use to authenticate themselves to one another without revealing their identities. A common assumption in image deduplication is that the server must be able to know which users own which images in order to identify duplicates between different users, so we follow the precedent set by \cite{Liu2015,Li2015,gang2015secure}, and assume brute-force server inference attacks to be outside our threat model.

\subsection{Adding Images}

The server's hash tables can only hold up to $2^h$ entries each, where $h$ is the size in bits of the result of the SLSH. When taking into account the desire to avoid collisions due to load, the practical upper bound is even lower. In other applications, the server can rehash its elements into a larger table when the number of elements it stores approaches that threshold. However, because the server cannot generate SLSHes (it does not have the original image or feature vector), it would have to ask the users to generate new hashes. This is costly to the users computationally, and is thus not a desirable approach.

Instead, the server can initialize a new set of hash tables $HT_1' \cdots HT_t'$, with a new set of parameters $p_1' \cdots p_t'$. Users uploading will henceforth provide the server two sets of hashes of their images' feature vectors: one set for the parameters $p_1 \cdots p_t$ of the original hash tables $HT_1 \cdots HT_t$, and one set for the parameters $p_1' \cdots p_t'$ of the new hash tables $HT_1' \cdots HT_t'$. Newly uploaded images are queried against all the hash tables, but only stored (if not deduplicated) in the new set. While this doubles the amount of computation the users must perform, these calculations are still only performed once, at image upload. Further, this strategy allows the server to store images beyond the original capacity of $HT_1 \cdots HT_t$ without violating user privacy. This scenario should be rare as long as $h$ is chosen to be sufficiently large, so that tables are not filled quickly and adding tables occurs only rarely.

\subsection{Sharing the Load}

Our system offers a high degree of privacy and functionality to its users, at the cost of extra computation. One of these costs is the PAKE-based key exchange that users must perform. The original owners of images that are ``popular'' (frequently selected for deduplication) bear a disproportionate part of this load.
A server can attempt to prevent this unfair situation by not always selecting the image's original owner to perform key exchange with new uploaders, but by instead selecting from all users who already have access and thereby distributing the load fairly. In this way, the ability of the server to infer which parties have uploaded similar images actually becomes an advantage for ensuring fairness among users. 

\subsection{Brute-Force Attacks}

In our protocol, both servers and clients can carry out brute-force attacks by repeatedly querying images against the server's storage to see if another client has stored a similar image with the server. Such an attack from the server cannot be theoretically prevented without introducing more assumptions (i.e. an extra server \cite{stanek2014secure}). The practical approach of rate-limiting user queries can prevent such attacks from users \cite{Li2015}.  

\subsection{Leveraging Trusted Hardware}
In order to prevent an adversary from conducting a brute-force attack, this protocol could be modified to utilize secure hardware to prevent a server from guessing how similar two images are by observing the number of hashes that match. For instance, using Intel SGX \cite{costan2016intel}, we could define a function that computes the similarity score in a secure enclave and only outputs a binary value to indicate whether or not the similarity score is above a threshold. This would prevent a malicious user from learning extra information regarding how exactly similar their image is to another user's, but would make the protocol hardware-dependent. Remote attestation can  securely verify that the server is running authenticated code that has not been tampered with.

\section{Experimental Evaluation}

\subsection{Testing Implementation}

We implemented and tested feature extraction, dimensionality reduction, nearest-neighbor searching, and the SPAKE2 subprotocol \cite{abdalla2005simple}. Our implementation of SPAKE2 is in C++, and uses GMP for algebraic operations \cite{Granlund12}. The other tests are written in Python, and make use of the OpenCV library for image processing \cite{opencv_library}.  Keras and Tensorflow \cite{abadi2016tensorflow,chollet2015keras} are used for feature extraction, and a modified version of lshash incorporating the SHA256 cryptographic hash 
was used for dimensionality reduction and nearest-neighbor search\cite{lshash}. To measure realistic performance, our tests were run on a server node belonging to a cluster in active use by a university (Intel Xeon CPUs, 128 GB of RAM, and GTX 1080Ti). The nodes  were not exclusively used by us, and our tests were run in an environment similar to servers under high load. This may have introduced extra latency and variance in our results.

We used training data from the standard image datasets featured in the Visual Decathalon Challenge \cite{rebuffi2017learning}. We have omitted results from the Imagenet dataset from our graphs for readability, though those results were also considered in drawing our conclusions. 
The number of images in each dataset is given in Table \ref{table:query}.

\begin{table}[t]
\centering
\caption{Time to Query 100 Images}\vspace{-5pt}
\label{table:query}
\begin{tabular}{|p{1.6cm}|p{0.9cm}|p{1.8cm}|p{2.3cm}|}
\hline\textbf{Dataset}  & \textbf{Size}  & \textbf{Avg. Time per Query (ms)} & \textbf{Avg. Time per Indexed Image (ms)}\\
\hline
vgg-flowers & 1020 & 8.648421288 & 0.847884440\\
\hline
dtd  & 1880 & 17.41103601 & 0.926118937 \\
\hline
aircraft  & 3334 & 15.30001068 & 0.458908539\\
\hline
ucf101 & 7585 & 34.59270334 & 0.456067282 \\
\hline
omniglot & 17853 & 48.57958031 & 0.272108779\\
\hline
daimlerpedcls & 23520 & 2.125451088 & 0.009036782 \\
\hline
gtsrb & 31367 & 95.60348845 & 0.304790029  \\
\hline
cifar100 & 40000 & 166.0737145 & 0.415184286  \\
\hline
svhn & 42566 & 105.8770444 & 0.248736185  \\
\hline
imagenet & 1232167 & 189.7144794 & 0.001897144 \\
\hline
\textbf{Average} & \textbf{140129} & \textbf{54.91238334} & \textbf{0.437648362} \\
\hline
\end{tabular}\vspace{-10pt}

\end{table}

\begin{table}[t]
\centering
\caption{Hash Computation Time for 24 bits (ms)}\vspace{-5pt}
\label{table:hash_tables}
\begin{tabular}{|p{1.4cm}|p{1.0cm}|p{0.6cm}|p{0.7cm}|p{0.8cm}|p{1.1cm}|}
\hline
\textbf{\# of Tables} & \textbf{Min} & \textbf{Avg.} & \textbf{Max} & \textbf{Median} & \textbf{Std. Dev.} \\
\hline
1 & 0.14 & 0.33 & 14.93 & 0.15 & 0.67 \\
\hline
2 & 0.24 & 0.26 & 0.32 & 0.26 & 0.01 \\
\hline
4 & 0.48 & 0.51 & 0.55 & 0.51 & 0.01 \\
 \hline
8 & 1.04 & 1.65 & 30.70 & 1.14 & 1.91 \\
 \hline
16 & 2.05 & 2.16 & 6.07 & 2.13 & 0.29 \\
\hline
32 & 4.08 & 14.85 & 214.00 & 12.79 & 17.33 \\
\hline
\end{tabular}\vspace{-10pt}

\end{table}

\begin{table}[t]
\centering
\caption{Hash Computation Time for 6 Tables (ms)}\vspace{-5pt}
\label{table:hash_bits}
\begin{tabular}{|l|p{1.0cm}|p{0.6cm}|p{0.7cm}|p{0.8cm}|p{1.1cm}|}
\hline
\textbf{Hash Length} & \textbf{Min} & \textbf{Avg.} & \textbf{Max} & \textbf{Median} & \textbf{Std. Dev.} \\
\hline
16 bits & 0.75 & 0.89 & 30.98 & 0.87 & 0.55 \\
\hline
24 bits & 0.77 & 0.93 & 33.75 & 0.81 & 0.76 \\
\hline
32 bits & 0.78 & 4.23 & 84.55 & 0.87 & 7.56 \\
\hline
64 bits & 0.83 & 9.15 & 84.55 & 5.80 & 9.21 \\
\hline
\end{tabular}\vspace{-10pt}

\end{table}

\subsection{Efficiency}

\noindent \textbf{Feature Extraction:}
The time to extract features for each database using ResNet50 is shown in Fig. \ref{graph:feat_ext}. We ran 10 trials on each dataset (with the exception of Imagenet, which was tested 5 times). Our results show that feature extraction on a single image takes about 33 ms on average. This computational overhead for an image upload is a manageable amount for a client. 

\noindent \textbf{Dimensionality Reduction:}
The time to index a database of images is shown in Fig. \ref{graph:indexing}. We performed 10 trials on each dataset (Imagenet was tested only 5 times). Our results show that indexing with 6 hash tables and a locality-sensitive hash size of 24 bits takes about 39 ms per image on average, taking hash calculation into account. These parameters were chosen to strike a balance between efficiency and accuracy. The computation time for a client is then even less, as they only need to calculate the hashes, and do not have to index the values into multiple hash tables.

In Tables \ref{table:hash_tables} and \ref{table:hash_bits} we show the time needed to calculate a client's hashes when varying the number of tables and hash size, with data from 4000 trials in each case. In particular, calculating a 24-bit hash for 6 tables takes 0.93 ms on average. As expected, the runtime for hash calculations increases linearly with both the number of tables and hash size. From our experiments we can thus conclude that both client-side hashing and server-side indexing are feasible and scalable.

\noindent \textbf{Nearest-Neighbor Searching:}
We tested the time for querying a small constant number (100) of images against each database, using 10 trials. We used the same index specifications as above. The resulting runtimes are shown in Table \ref{table:query}, which includes both the average time per image query and average query time divided by database size. We can conclude that the average time to query per image may be as little as 2.13 ms. The average time per query across all tested datasets was about 55 ms. Interestingly, we note that the average time for a query does not increase as the size of the previously indexed dataset does, and may even decrease. A possible explanation is that cache/memory coherency yields greater benefits for queries over larger databases. From this, we conclude that querying is computationally feasible and also scalable.


\noindent \textbf{Access Control:}
Our implementation of PAKE was tested over cyclic groups with prime orders of 1024, 2048, 4096, and 8192 bits. Each group was tested with 1000 trials. Even for group sizes of 8192 bits, the user computation averaged below 170 $\mu$s. The time to perform user computation is not dependent on any other parameters of the protocol. This shows that the use of PAKE for key exchange is feasible.

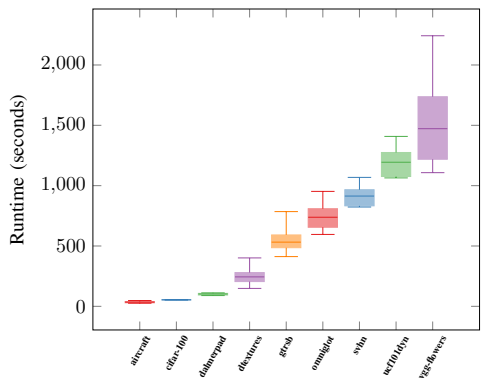
\begin{figure}[t]  
\centering
\begin{tikzpicture}[scale=0.75]
\pgfplotstablegetrowsof{\datatable}
\pgfmathtruncatemacro{\rownumber}{\pgfplotsretval-1}
\begin{axis}[boxplot/draw direction=y,
xticklabels={aircraft,cifar-100,dalmerpad,dtextures,gtrsb, omniglot, svhn,ucf101dyn,vgg-flowers}, xtick={1,...,\the\numexpr\rownumber+1},
x tick label style={scale=0.5,font=\bfseries, rotate=60,align=center},
ylabel={Runtime (seconds)},cycle list/Set1-5 ]

\pgfplotsinvokeforeach{0,...,\rownumber}{

    \pgfplotstablegetelem{#1}{min}\of\datatable
    \edef\mymin{\pgfplotsretval}

    \pgfplotstablegetelem{#1}{avg}\of\datatable
    \edef\myavg{\pgfplotsretval}

    \pgfplotstablegetelem{#1}{max}\of\datatable
    \edef\mymax{\pgfplotsretval}

    \pgfplotstablegetelem{#1}{median}\of\datatable
    \edef\mymedian{\pgfplotsretval}

    \pgfplotstablegetelem{#1}{deviation}\of\datatable
    \edef\mydeviation{\pgfplotsretval}

    \typeout{\mymin,\mymax,\myavg,\mymedian,\mydeviation}
    \pgfmathsetmacro{\mylowerq}{\mymedian-\mydeviation}
    \pgfmathsetmacro{\myupperq}{\mymedian+\mydeviation}
    \edef\temp{\noexpand\addplot+[,
        boxplot prepared={
            lower whisker=\mylowerq,
            upper whisker=\mymax,
            lower quartile=\mymin,
            upper quartile=\myupperq,
            median=\myavg,
            every box/.style={solid,fill,opacity=0.5},
            every whisker/.style={solid },
            every median/.style={solid},
        }, 
        ]coordinates {};}
    \temp
}

\end{axis}
\end{tikzpicture}\vspace{-10pt}
\caption{Feature Extraction Runtime}\vspace{-10pt}
\label{graph:feat_ext}
\end{figure}

\begin{figure}[t]  
\centering
\begin{tikzpicture}[scale=0.75]
\pgfplotstablegetrowsof{\indextable}
\pgfmathtruncatemacro{\rownumber}{\pgfplotsretval-1}
\begin{axis}[boxplot/draw direction=y,
xticklabels={aircraft,cifar-100,daimlerped,dtextures,gtrsb, omniglot, svhn,ucf101dyn,vgg-flowers}, xtick={1,...,\the\numexpr\rownumber+1},
x tick label style={scale=0.5,font=\bfseries, rotate=60,align=center},
ylabel={Runtime (seconds)},cycle list/Set1-5 ]

\pgfplotsinvokeforeach{0,...,\rownumber}{

    \pgfplotstablegetelem{#1}{min}\of\indextable
    \edef\indexmin{\pgfplotsretval}

    \pgfplotstablegetelem{#1}{avg}\of\indextable
    \edef\indexavg{\pgfplotsretval}

    \pgfplotstablegetelem{#1}{max}\of\indextable
    \edef\indexmax{\pgfplotsretval}

    \pgfplotstablegetelem{#1}{median}\of\indextable
    \edef\indexmedian{\pgfplotsretval}

    \pgfplotstablegetelem{#1}{deviation}\of\indextable
    \edef\indexdeviation{\pgfplotsretval}

    \typeout{\indexmin,\indexmax,\indexavg,\indexmedian,\indexdeviation}
    \pgfmathsetmacro{\indexlowerq}{\indexmedian-\indexdeviation}
    \pgfmathsetmacro{\indexupperq}{\indexmedian+\indexdeviation}
    \edef\temp{\noexpand\addplot+[,
        boxplot prepared={
            lower whisker=\indexlowerq,
            upper whisker=\indexmax,
            lower quartile=\indexmin,
            upper quartile=\indexupperq,
            median=\indexavg,
            every box/.style={solid,fill,opacity=0.5},
            every whisker/.style={solid },
            every median/.style={solid},
        }, 
        ]coordinates {};}
    \temp
}

\end{axis}
\end{tikzpicture}\vspace{-10pt}
\caption{Index Construction Runtime}\vspace{-10pt}
\label{graph:indexing}
\end{figure}
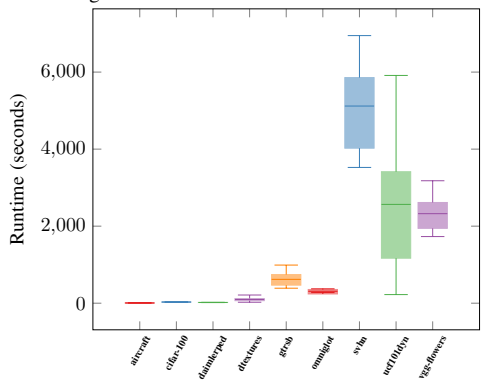

\newcommand{\distfigw}{.32}
\begin{figure*}[t]
    \centering
    \begin{tabular}{ccc}
      \includegraphics[width=\distfigw\textwidth]{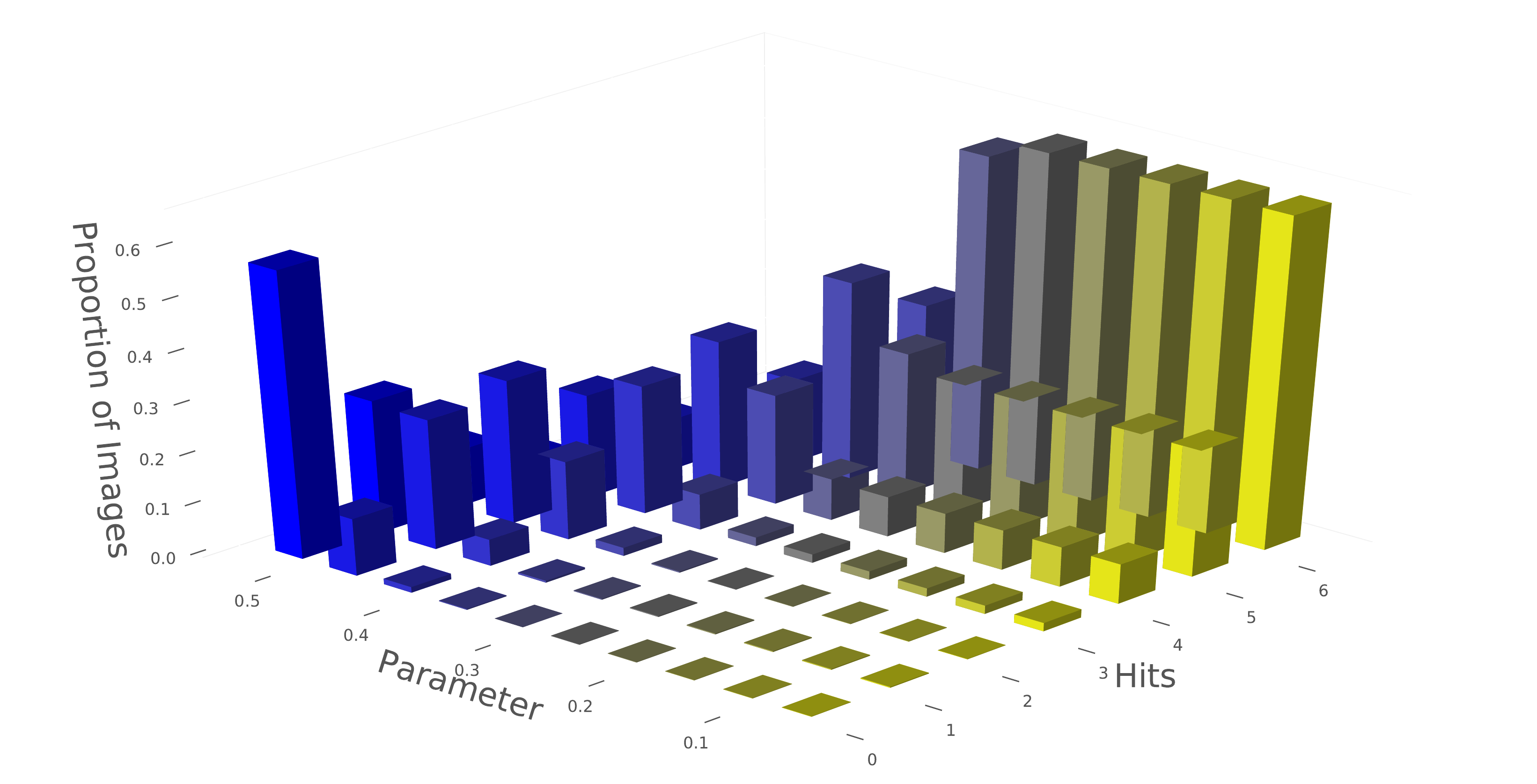}   &  \includegraphics[width=\distfigw\textwidth]{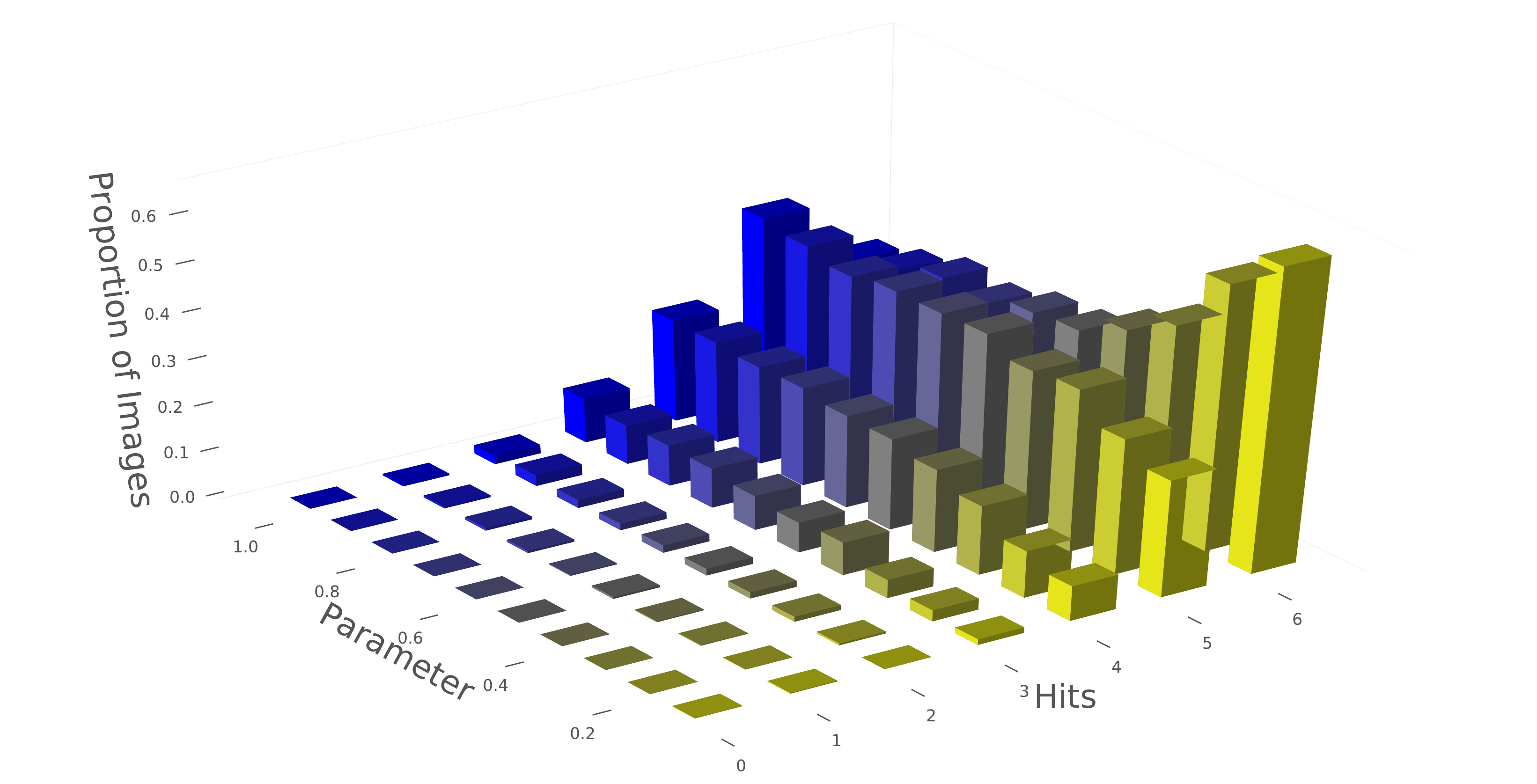} & 
      \includegraphics[width=\distfigw\textwidth]{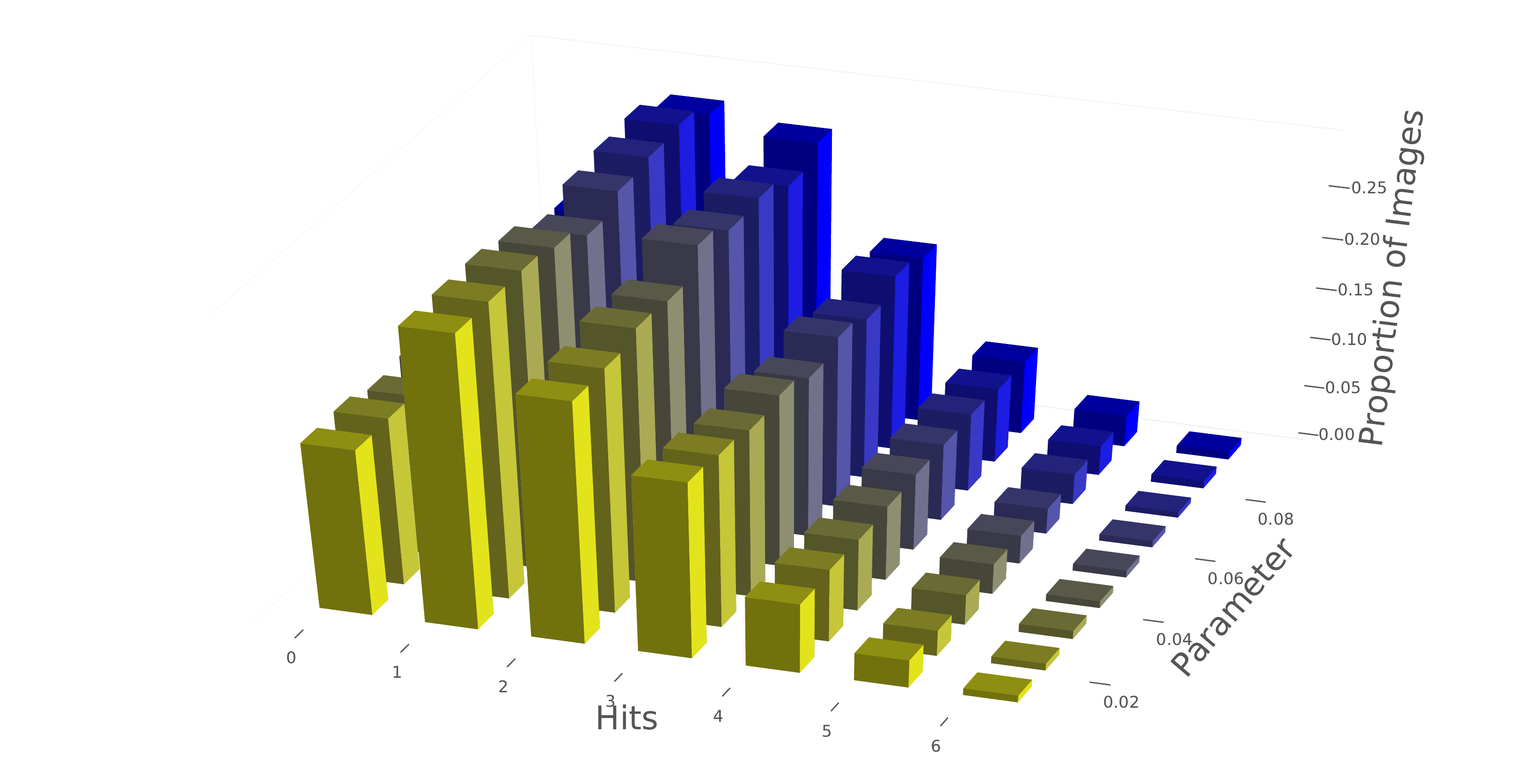}\\
        (a) Blur & (b) Brighten  &  (c) Gaussian noise\\
        \includegraphics[width=\distfigw\textwidth]{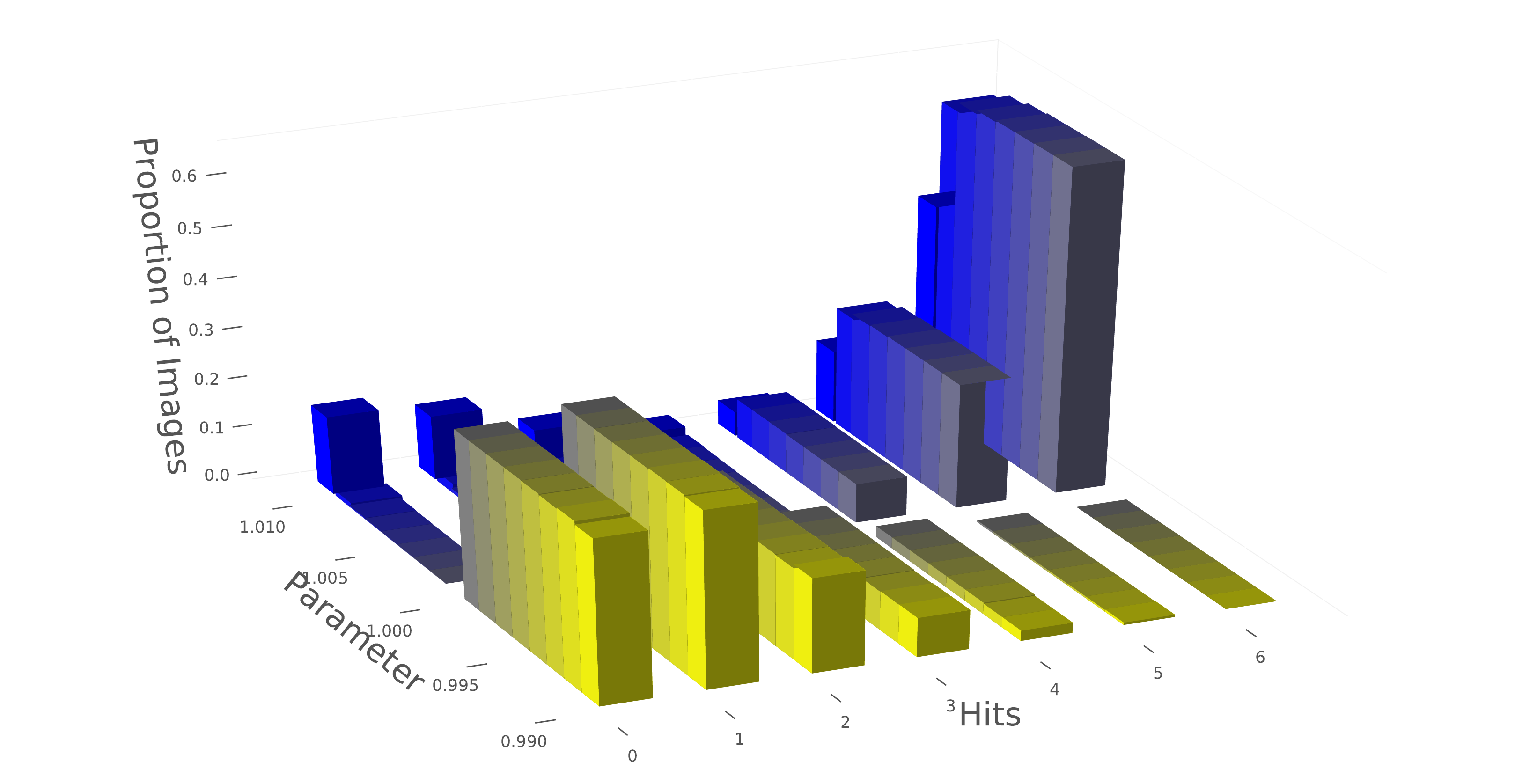} &
        \includegraphics[width=\distfigw\textwidth]{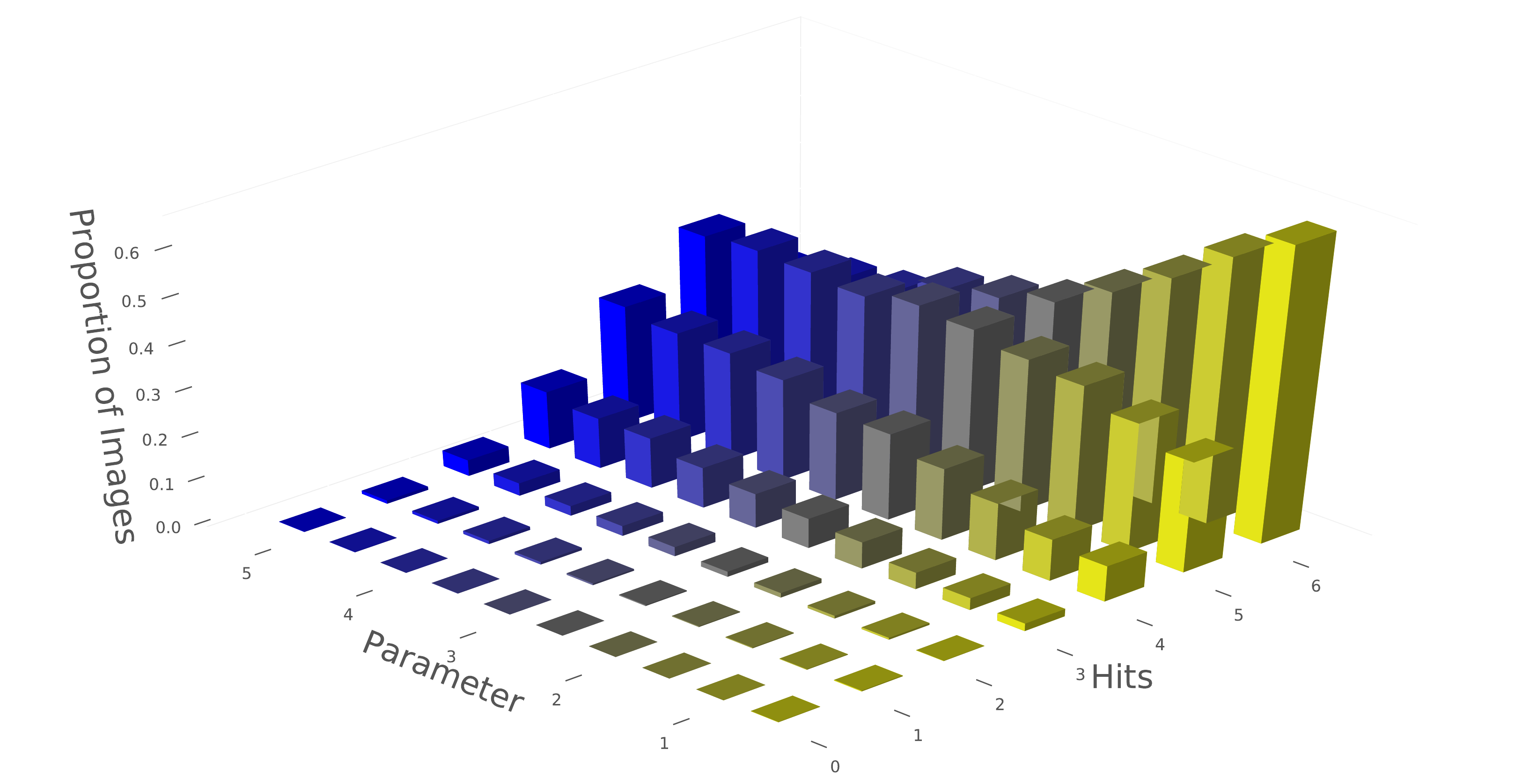}&
        \includegraphics[width=\distfigw\textwidth]{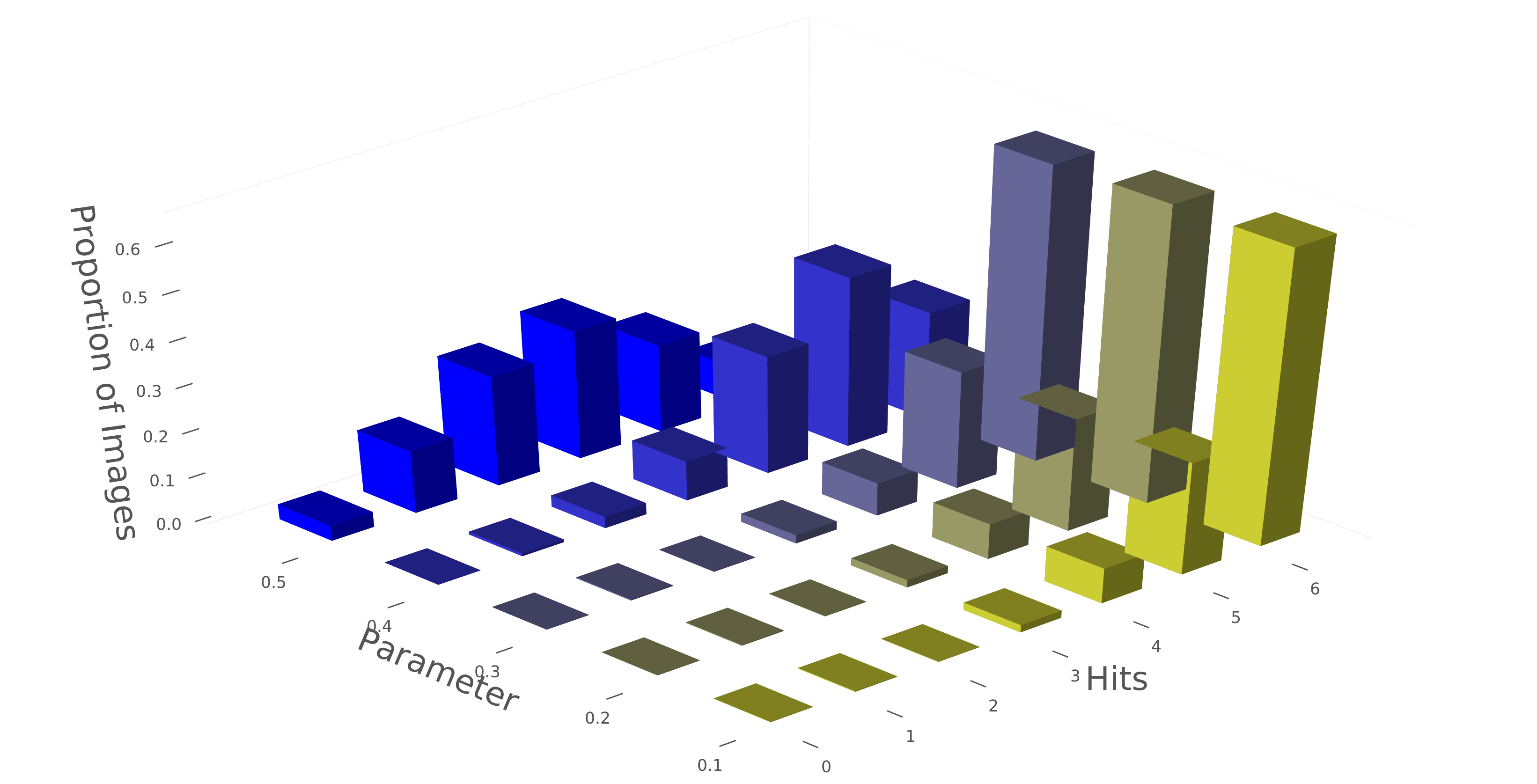} \\
        (d) Resize & (e) Saturate & (f) Sharpen
    \end{tabular}
    \begin{tabular}{cc}
        \includegraphics[width=\distfigw\textwidth]{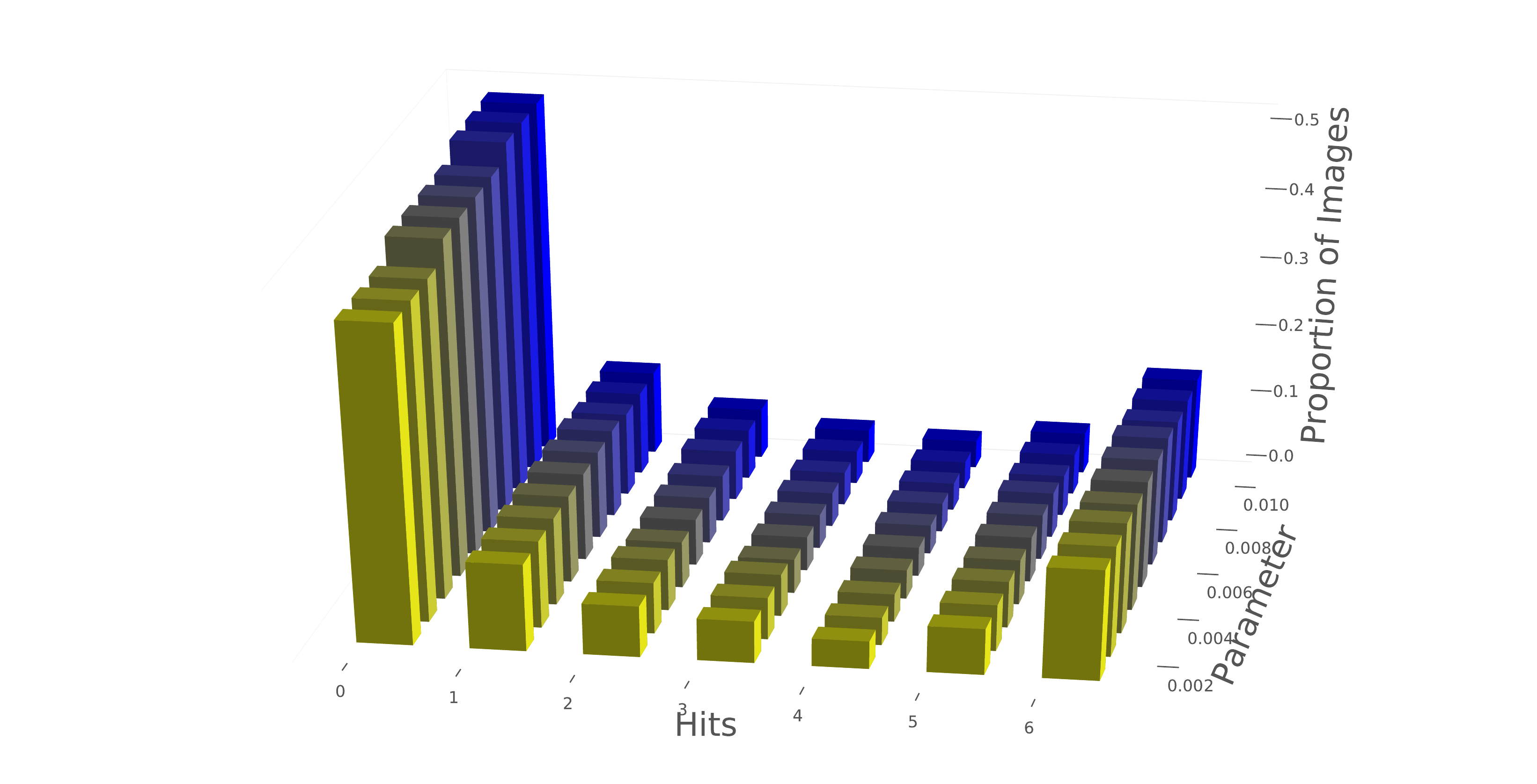}&
        \includegraphics[width=\distfigw\textwidth]{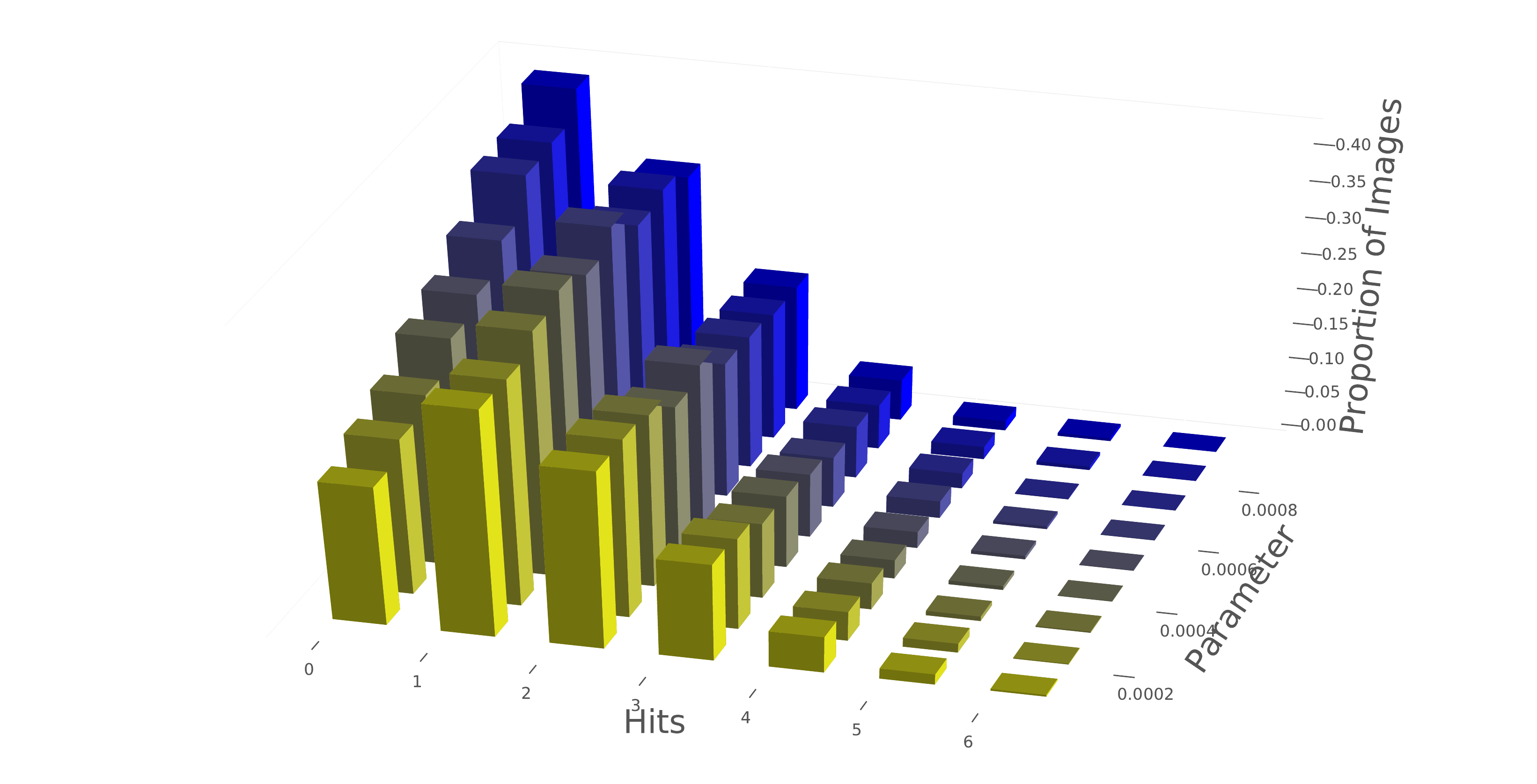}
        \\
        (g) Solarize & (h) Salt/Pepper noise \\
    \end{tabular}
    \caption{Various distortions of different strengths and number of matches in the hash tables. (Different angles are used to show the efficacy or failure in detecting similarity from slight distortions.)}
    \label{fig:accuracy}
\end{figure*}

\subsection{Distortions' Impacts to Our Deduplication}

We tested the propensity of our nearly-identical deduplication scheme to identify images as similar after small distortions are applied. We randomly chose a subset of Imagenet, and applied gradually increasing distortions to images in that subset. We then ran queries with those images and observed how many hash tables recorded a match with the original image.
The results are shown in Figures \ref{fig:accuracy}(a)-(h). In these graphs, the proportion of queries with some number of matches is shown as a function of the number of matches and the severity of the distortion. For example, in Figure \ref{fig:accuracy}(a), the figure shifts from yellow to blue as query images become more blurry, and there is a visible trend of the number of matches decreasing as the distortion increases.
These results show that our system is able to accurately identify (with $c = \lceil \frac{t+1}{2} \rceil$, i.e. hash collisions in more than half the tables) similar images with small changes from blurring, brightening, enlargement, saturation, and sharpening. The system was not able to reliably detect nearly-identical images with distortions of solarization or salt-and-pepper noise, and performed somewhat poorly with Gaussian noise - this is logical, as those types of distortions will affect features more. Shrinking the image also resulted in poor performance, which makes sense, as shrinking an image results in a loss of information.
Our system performed extremely well for false positives (i.e. an image not the original identified as similar) - none of our tests had more than one table indicate a false positive. 

\newcommand{\qosfigw}{.44}
\begin{figure*}[t]
    \centering
    \begin{tabular}{cc}
     \includegraphics[width=\qosfigw\linewidth]{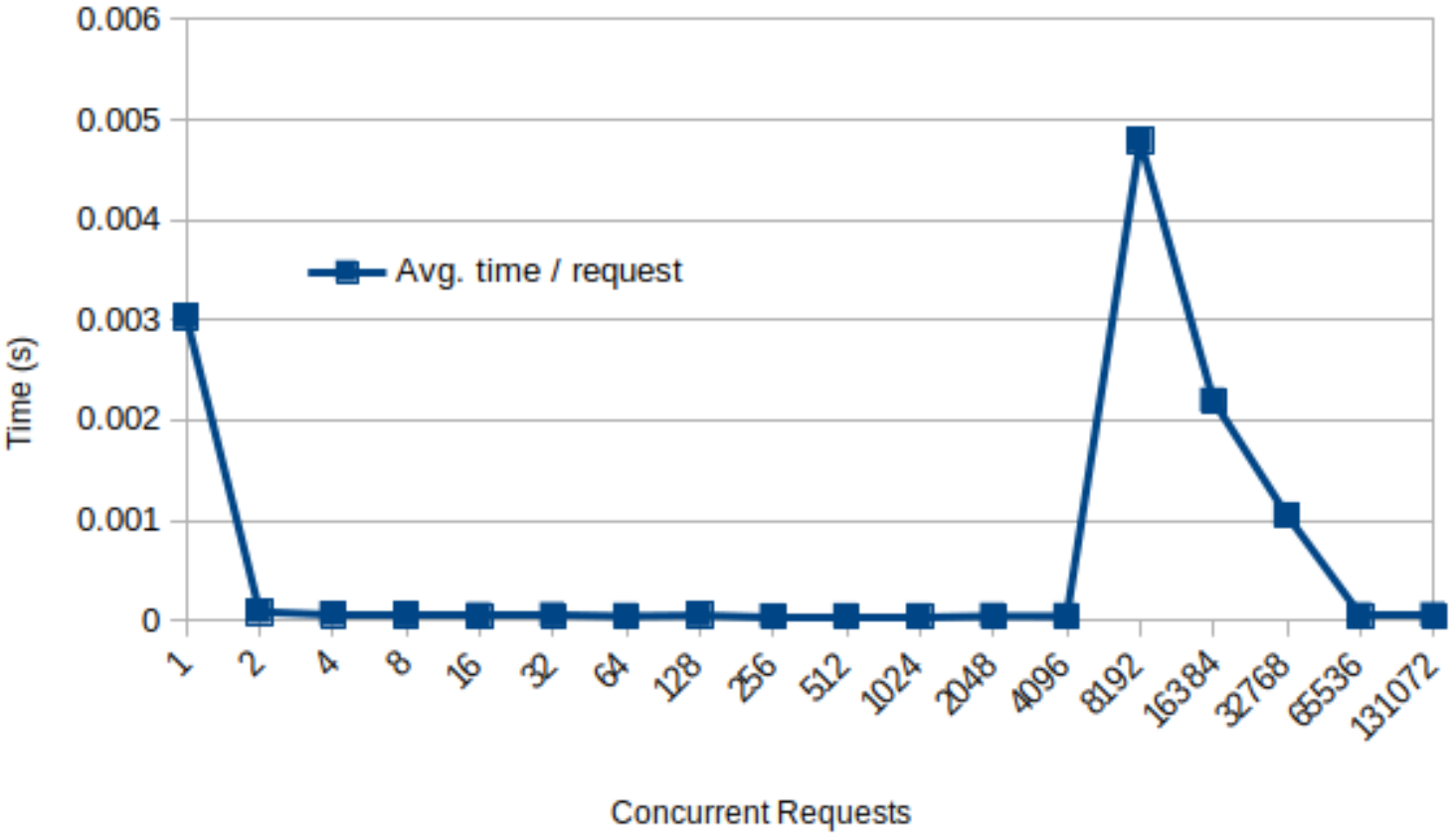}    &  \includegraphics[width=\qosfigw\linewidth]{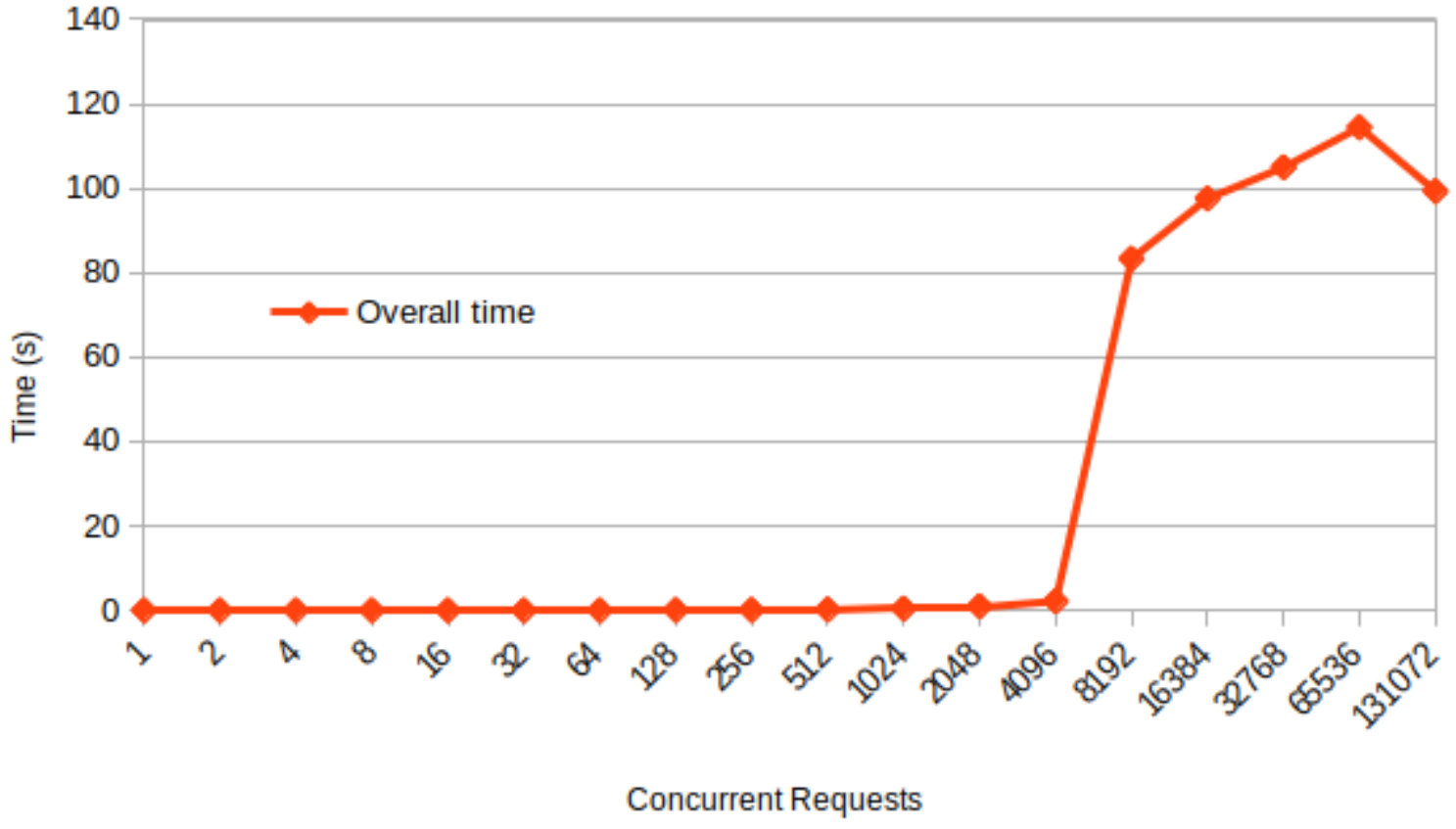}\\
      (a) Avg. time per query   & (b) Total query time  \\
       \includegraphics[width=\qosfigw\linewidth]{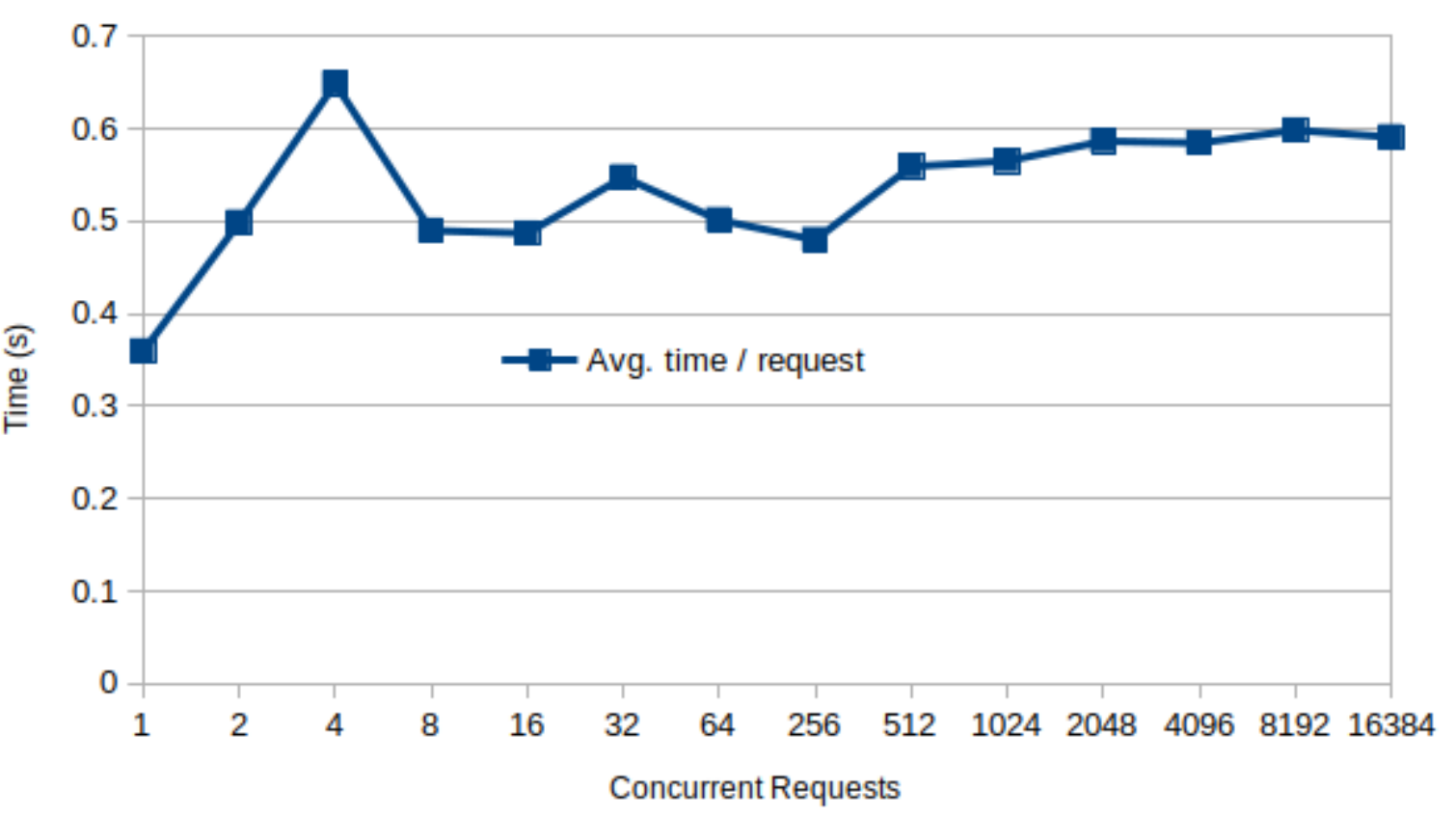}  & \includegraphics[width=\qosfigw\linewidth]{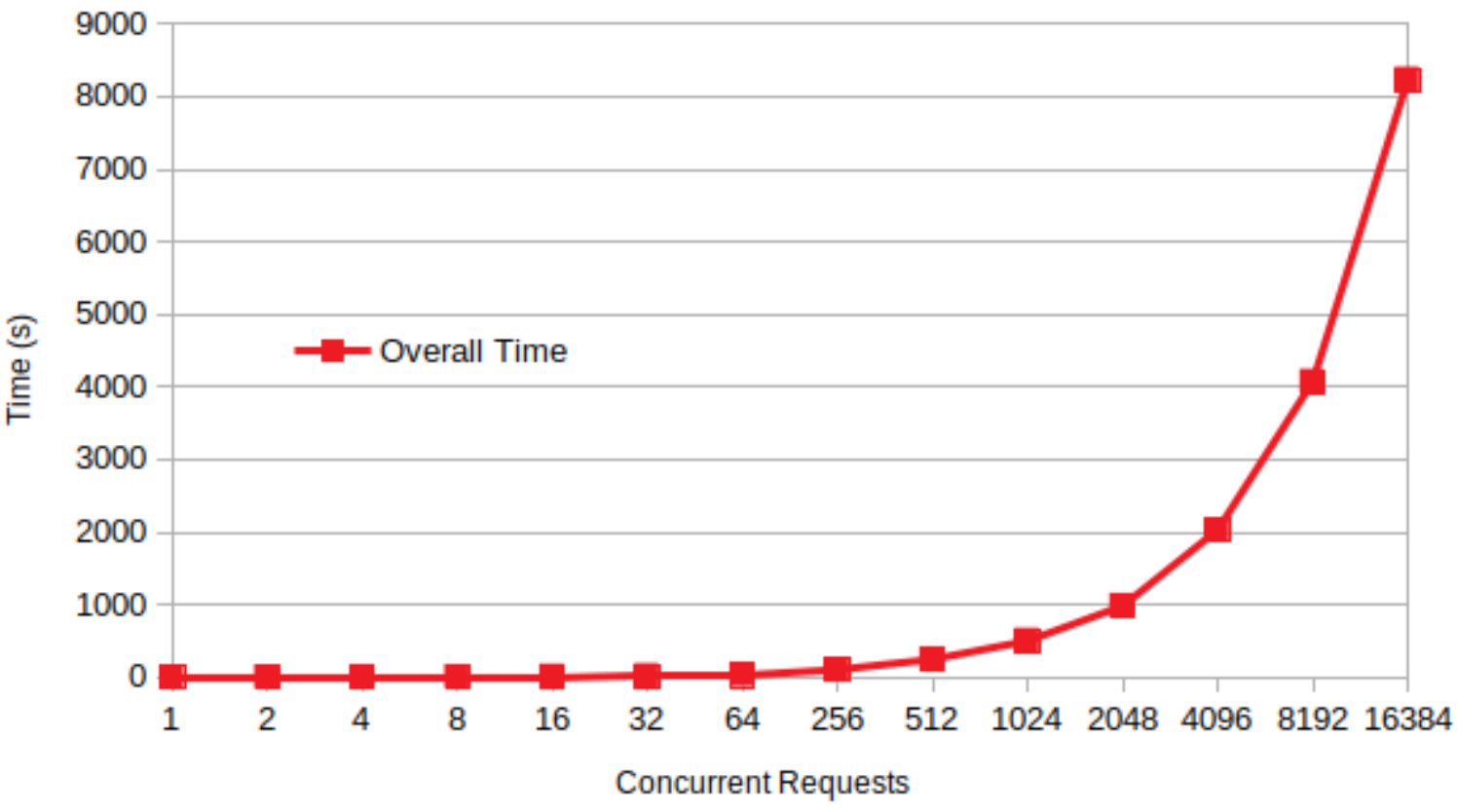} \\
       (c) Avg. time to index  &  (d) Total index time \\
    \end{tabular}
    \caption{Time to handle concurrent requests.}
    \label{fig:qos}
\end{figure*}

\subsection{Quality of Service with Concurrent Requests}

First, we examined how well our system could respond to multiple simultaneous queries. We measured the runtimes of each individual request, as well as the total runtime of the whole set of requests. The results are averaged over five trials, and used up to 16384 threads. The average time for only a single request (Figure \ref{fig:qos}(a)) was higher due to the overhead of initialization. For the rest of the runtimes up to 16384 requests, the average time was on order of 0.1 ms. The average times increase greatly as the number of requests grows close to 16384, as the overhead from more threads increases. After that point, the average time decreases again, as each thread will then have multiple requests. The total runtime for all of the requests (Figure \ref{fig:qos}(b)) shows that up to a certain level of saturation (around 8192 simultaneous requests), the overall runtime was small (hundredths of seconds up to 32 requests, and seconds or less for up to 4096 requests). This shows that for client queries, our protocol is efficient for many simultaneous requests. 

Next, we examined how our implementation handled simultaneous indexing of new images. The average request runtimes are shown in Figure \ref{fig:qos}(c), and the total time to index the entire set of new images is shown in Figure \ref{fig:qos}(d). The time for a single request to complete was under 0.65 seconds in all cases, showing that concurrent indexing is efficient, even with locking. The time to fulfill all requests increased linearly with the number of requests. Our implementation used simple locking. A more sophisticated database system might be able to allow more efficient indexing, though this is beyond the scope of our work.


\section{Conclusion}

This paper presents the first protocol for nearly-identical image deduplication with only a single untrusted server. Our rigorous proof shows the protocol's security in the highly challenging case of fully malicious and colluding adversaries. We also discuss practical issues widely applicable to deduplication. Finally, our experiments show the efficacy and efficiency of our protocol and its components.

\bibliographystyle{abbrv}
\bibliography{sample.bib}

\end{document}